\documentclass{llncs}

\renewcommand{\qed}{\hfill \ensuremath{\Box}}

\newtheorem{thm}{Theorem} 

\newtheorem{lem}{Lemma}

\newtheorem{prop}{Proposition}
\usepackage{amssymb,amsmath,pifont,footnote,epic,eepic,gastex,wasysym}
\usepackage{latexsym}
\usepackage{array}
\usepackage{stmaryrd}
\usepackage{algorithmic,algorithm}
\usepackage{subfigure}
\usepackage{color}

\newcommand{\Heading}[1]{\paragraph{\bf{#1}}}

\newcommand{\NOP}{\mathit{nop}}
\newcommand{\INC}{\mathit{inc}}
\newcommand{\DEC}{\mathit{dec}}

\newcommand{\GOTO}{\mathit{goto}}
\newcommand{\IF}{\mathit{if}}
\newcommand{\ELSE}{\mathit{else}}
\newcommand{\Sup}[1]{\overline{#1}}
\newcommand{\Inf}[1]{\underline{#1}}

\newcommand{\Q}{\mathbb{Q}}

\newcommand{\Avg}{\mathit{avg}}

\newcommand{\LimAvg}{\mathit{LimAvg}}
\newcommand{\LimInfAvg}{\mathit{LimInfAvg}}
\newcommand{\LimSupAvg}{\mathit{LimSupAvg}}

\begin{document}
\title{Robust Multidimensional Mean-Payoff Games are Undecidable}
\author{Yaron Velner}
\institute{The Blavatnik School of Computer Science, Tel Aviv University, Israel}
\maketitle

\begin{abstract}
Mean-payoff games play a central role in quantitative synthesis and verification.
In a single-dimensional game a weight is assigned to every transition and the objective of the protagonist is to assure a non-negative limit-average weight.
In the multidimensional setting, a weight vector is assigned to every transition and the objective of the protagonist is to satisfy a boolean condition over the limit-average weight of each dimension, e.g., $\LimAvg(x_1) \leq 0 \vee \LimAvg(x_2)\geq 0 \wedge \LimAvg(x_3) \geq 0$.
We recently proved that when one of the players is restricted to finite-memory strategies then the decidability of determining the winner is inter-reducible with Hilbert's Tenth problem over rationals (a fundamental long-standing open problem).
In this work we consider arbitrary (infinite-memory) strategies for both players and show that the problem is undecidable.
\end{abstract}

\pagestyle{headings}

\section{Introduction}
Two-player games on graphs provide the mathematical foundation for the study of reactive systems.
In these games, the set of vertices is partitioned into player-1 and player-2  vertices;
initially, a pebble is placed on an initial vertex, and in every round, the player who owns the vertex that the pebble resides in, advances the pebble to an adjacent vertex.
This process is repeated forever and give rise to a \emph{play} that induces an infinite sequence of edges.
In the quantitative framework, an objective assigns a value to every play, and the goal of player 1 is to assure a value of at least $\nu$ to the objective.
In order to have robust quantitative specifications, it is necessary to investigate games on graphs with multiple (and possibly conflicting) objectives.
Typically, multiple objectives are modeled by multidimensional weight functions (e.g.,~\cite{BrazdilBCFK11,BrazdilCKN12,ChatterjeeRR12,Alur:2009:ODM:1532848.1532880}), and the outcome of a \emph{play} is a vector of values $(r_1,r_2,\dots,r_k)$.
A \emph{robust specification} is a boolean formula over the atoms $r_i \sim \nu_i$, for $\sim\in\{\leq,<,\geq,>\}$, $i\in\{1,\dots,k\}$ and $\nu_i\in\Q$.
For example, $\varphi = ((r_1 \geq 9 \vee r_2 \leq 9) \wedge r_3 < 0 \wedge r_4 > 9)$.
The most well studied quantitative metric is the mean-payoff objective, which assigns the limit-average (long-run average) weight to an infinite sequence of weights (and if the limit does not exist, then we consider the limit infimum of the sequence).
In this setting, $r_i$ is the limit-average of dimension $i$ of the weight function, and the goal of player~1 is to satisfy the boolean condition.
In this work we prove that determining whether player~1 can satisfy such a condition is undecidable.

\Heading{Related work.}
The model checking problem (one-player game) for such objectives (with some extensions) was considered in~\cite{Alur:2009:ODM:1532848.1532880,mean-payoff-Automaton-Expressions,BokerCHK11,TomitaHHY12,Velner12} and decidability was established.
Two-player games for restricted subclasses that contain only conjunction of atoms were studied in~\cite{abs-1209-3234,ChatterjeeRR12,TACAS13,ChatterjeeV13} and tight complexity bounds were obtained (and in particular, the problem was proved to be decidable).
In~\cite{VelnerR11} a subclass that contains disjunction and conjunction of atoms of the form $r_i\sim\nu_i$ for $\sim\in\{\geq, >\}$ was studied and decidability was shown.
In~\cite{Velner14} we considered a similar objective but restricted player-1 to play only with finite-memory strategies.
We showed that the problem is provably hard to solve and its decidability is inter-reducible with Hilbert's tenth problem over rationals --- a fundamental long standing open problem.
In this work we consider for the first time games with robust quantitative class of specifications that is closed under boolean union, intersection and complement with arbitrary (infinite-memory) strategies.

Undecidability for (single-dimensional) mean-payoff games was proved for partial information mean-payoff games~\cite{DegorreDGRT10} and for mean-payoff games that are played over infinite-state pushdown automata~\cite{ChatterjeeV12}.
These works did not exploit the different properties of the $\geq$ and $\leq$ operators (which correspond to the different properties of limit-infimum-average and limit-supremum-average).
To the best of our knowledge, the undecidability proof in the paper is the first to exploit these properties. (As we mentioned before, when we consider only the $\geq$ and $>$ operators, the problem is decidable.)

\Heading{Structure of this paper.}
In the next section we give the formal definitions for robust multidimensional mean-payoff games.
We prove undecidability by a reduction from the halting problem of a two-counter machine.
For this purpose we first present a reduction from the halting problem of a one-counter machine and then we extend it to two-counter machine.
In Section~\ref{sec:Reduction} we present the reduction and give an intuition about its correctness.
In Section~\ref{sec:DetailedProof} we give a formal proof for the correctness of the reduction and extend the reduction to two-counter machine.
In Section~\ref{sec:Discuss} we discuss the key elements of our proof and apply the undecidability result for the similar model of mean-payoff expressions.
\section{Robust Multidimensional Mean-Payoff Games}\label{sec:RobustGames}
\smallskip\noindent{\bf Game graphs.} 
A \emph{game graph} $G=((V,E),(V_1,V_2))$ consists of a \emph{finite} 
directed graph $(V,E)$ with a set of vertices $V$ a set of edges $E$, 
and a partition $(V_1,V_2)$ of $V$ into two sets.
The vertices in $V_1$ are {\em player-1 vertices}, where player~1 chooses the
outgoing edges, and the vertices in $V_2$ are {\em player~2 vertices},
where  player~2 (the adversary to player~1) chooses the outgoing edges.
We assume that every vertex has at least one out-going edge.

\smallskip\noindent{\bf Plays.} A game is played by two players: 
player~1 and player~2, who form an infinite path in the game graph by 
moving a token along edges.
They start by placing the token on an initial vertex, and then they
take moves indefinitely in the following way.
If the token is on a vertex in~$V_1$, then player~1 moves the token along
one of the edges going out of the vertex.
If the token is on a vertex in~$V_2$, then player~2 does likewise.
The result is an infinite path in the game graph, called {\em plays}.
Formally, a \emph{play} is an infinite sequence of vertices such that $(v_k,v_{k+1}) \in E$ for all $k \geq 0$. 

\smallskip\noindent{\bf Strategies.} 
A strategy for a player is a rule that specifies how to extend plays.
Formally, a \emph{strategy} $\tau$ for player~1 is a function 
$\tau$: $V^* \cdot V_1 \to V$ that, given a finite sequence of vertices 
(representing the history of the play so far) which ends in a player~1 
vertex, chooses the next vertex.
The strategy must choose only available successors.
The strategies for player~2 are defined analogously.

\smallskip\noindent{\bf Multidimensional mean-payoff objectives.} 
For multidimensional mean-payoff objectives we will consider game graphs 
along with a weight function $w: E \to \Q^k$ that maps each edge to a vector
of rational weights.
For a finite path $\pi$, we denote by $w(\pi)$ the sum of the weight vectors 
of the edges in $\pi$ and $\Avg(\pi) = \frac{w(\pi)}{|\pi|}$, 
where $|\pi|$ is the length of $\pi$, denote the average vector of the 
weights.
We denote by $\Avg_i(\pi)$ the projection of $\Avg(\pi)$ to the $i$-th 
dimension.
For an infinite path $\pi$, let $\pi_i$ denote the finite prefix of 
length $i$ of $\pi$; and we define
$\LimInfAvg_i(\pi) =\lim\inf_{i \to\infty} \Avg(\rho_i)$ and analogously 
$\LimSupAvg_i(\pi)$ with $\lim\inf$ replaced by $\lim\sup$.
For an infinite path $\pi$, we denote by 
$\LimInfAvg(\pi)= (\LimInfAvg_1(\pi),\dots,\LimInfAvg_k(\pi))$ 
(resp. $\LimSupAvg(\pi) = (\LimSupAvg_1(\pi),\dots,\LimSupAvg_k(\pi))$)  
the limit-inf (resp. limit-sup) vector of the averages (long-run average or 
mean-payoff objectives).
A multidimensional mean-payoff condition is a boolean formula over the atoms $\LimInfAvg_i \sim \nu_i$ for $\sim\in\{\geq, >,\leq , >\}$.
For example, the formula $\LimInfAvg_1 > 8 \vee \LimInfAvg_2 \leq -10 \wedge \LimInfAvg < 9$ is a possible condition and a path $\pi$ satisfies the formula if $\LimInfAvg_1(\pi) > 8 \vee \LimInfAvg_2(\pi) \leq -10 \wedge \LimInfAvg(\pi) < 9$.
We note that we may always assume that the boolean formula is positive (i.e., without negation), as, for example, we can always replace $\neg (r \geq \nu)$ with $r < \nu$.

For a given multidimensional weighted graph and a multidimensional mean-payoff condition, we say that player~1 is the winner of the game if he has a winning strategy that satisfy the condition against any player-2 strategy.

For an infinite sequence or reals $x_1,x_2,x_3,\dots$ we have $\LimInfAvg(x_1,x_2,\dots) = -\LimSupAvg(-x_1,-x_2,\dots)$.
Hence, an equivalent formulation for multidimensional mean-payoff condition is a positive boolean formula over the atoms $\LimInfAvg_i \sim \nu_i$ and $\LimSupAvg_i \sim \nu_i$ for $\sim\in\{\geq, >\}$.
For positive formulas in which only the $\LimInfAvg_i \sim\nu_i$ occur, determining the winner is decidable by~\cite{VelnerR11}.
In the sequel we abbreviate $\LimInfAvg_i$ with $\Inf{i}$ and $\LimSupAvg_i$ with $\Sup{i}$.
In this work we prove undecidability for the general case and for this purpose it is enough to consider only the $\geq$ operator and thresholds $0$.
Hence, in the sequel, whenever it is clear that the threshold is $0$, we abbreviate the condition $\Inf{i} \geq 0$ with $\Inf{i}$ and $\Sup{i} \geq 0$ with $\Sup{i}$.
For example, $\Inf{i}\vee \Sup{j} \wedge \Inf{\ell}$ stands for $\LimInfAvg_i \geq 0 \vee \LimSupAvg_j \geq 0 \wedge \LimInfAvg_\ell \geq 0$.
By further abuse of notation we abbreviate the current total weight in dimension $i$ by $i$ (and make sure that the meaning of $i$ is always clear from the context) and the absolute value of the total weight by $|i|$.

\section{Reduction from the Halting Problem and Informal Proof of Correctness}\label{sec:Reduction}
In this work we prove the undecidability of determining the winner in games over robust mean-payoff condition by a reduction from the halting problem of two-counter machine.
For this purpose we will first show a reduction from the halting problem of a one-counter machine to robust mean-payoff games, and the reduction from two-counter machines relays on similar techniques.
We first give a formal definition for a one-counter machine, and in order to simplify the proofs we give a non-standard definition that is tailored for our needs.   
A \emph{two-sided one-counter machine} $M$ consists of two finite set of control states, namely $Q$ (\emph{left states}) and $P$ (\emph{right states}), an initial state $q_0\in Q$, a final state $q_f\in Q$, a finite set of \emph{left to right} instructions $\delta_{\ell\to r}$ and a finite set of \emph{right to left} instructions $\delta_{r\to\ell}$.
An instruction determines the next state and manipulates the value of the counter $c$ (and initially the value of $c$ is $0$). 
A left to right instruction is of the form of either:
\begin{itemize}
\item $q: \IF c=0\ \GOTO\ p\ \ELSE\ c:= c - 1\ \GOTO\ p'$, for $q\in Q$ and $p,p'\in P$ ; or
\item $q: \GOTO\ p$, for $q\in Q$ and $p\in P$ (the value of $c$ does not change).
\end{itemize}
A right to left instruction is of the form of either
\begin{itemize}
\item $p: c:= c+1\ \GOTO\ q$, for $p\in P$ and $q\in Q$ ; or
\item $p: \GOTO\ q$, for a state $p\in P$ and a state $q\in Q$ (the value of $c$ does not change).
\end{itemize}
We observe that in our model, decrement operations are allowed only in left to right instructions and increment operations are allowed only in right to left instructions.
However, since the model allows state transitions that do not change the value of the counter (\emph{nop transitions}), it is trivial to simulate a standard one-counter machine by a two-sided counter machine.

For the reduction we use the states of the game graph to simulate the states of the counter machine and we use two dimensions to simulate the value of the one counter.
In the most high level view our reduction consists of three main gadgets, namely, reset, sim and blame (see Figure~\ref{fig:Overview}), and a state $q_f$ that represents the final state of the counter machine.
Intuitively, in the sim gadget player~1 simulates the counter machine, and if the final state $q_f$ is reached then player~1 loses.
If player~2 detects that player~1 does not simulate the machine correctly, then the play goes to the blame gadget.
From the blame gadget the play will eventually arrive to the reset gadget.
This gadget assign proper values for all the dimensions of the game that are suited for an honest simulation in the sim gadget.


\begin{figure}
\begin{center}
\begin{picture}(50,30)(0,0) 



\node[Nmarks=i,Nw=22,Nh=8,Nmr=3](RESET)(0,0){reset}
\rmark(RESET)

\node[Nw=22,Nh=8,Nmr=3](SIM)(35,25){sim}
\rmark(SIM)

\node[Nw=22,Nh=8,Nmr=3](BLAME)(35,0){blame}
\rmark(BLAME)

\node(Qf)(60,25){$q_f$}

\drawedge[sxo=-5.625,exo=-5.625](RESET,SIM){}
\drawedge(SIM,Qf){}
\drawedge(BLAME,RESET){}

\drawloop[loopangle=270](Qf){$x\gets -1$}

\drawedge[sxo=-5.625,exo=-5.625](SIM,BLAME){}
\drawedge[sxo=-1.875,exo=-1.875](SIM,BLAME){}
\drawedge[sxo=1.875,exo=1.875](SIM,BLAME){}
\drawedge[sxo=5.625,exo=5.625](SIM,BLAME){}

\end{picture}
\end{center} 
\caption{Overview.}\label{fig:Overview}
\end{figure}
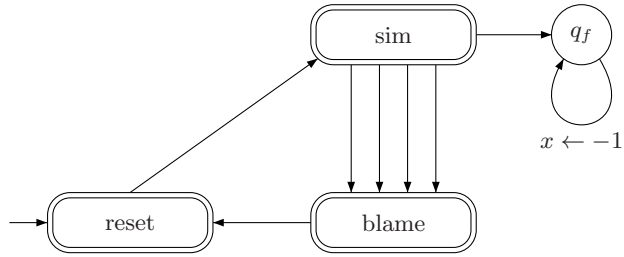


We now describe the construction with more details.
We first present the winning objective and then we describe each of the three gadgets.
For a two-sided counter machine $M$ we construct a game graph with 8 dimensions denoted by $\ell,r,g_s,c_+,c_-,g_c, x$ and $y$ and the objective
\[(\Inf{\ell} \wedge \Inf{r} \vee\Sup{g_s})\wedge (\Inf{c_+}\wedge \Inf{c_-} \vee \Sup{g_c})\wedge \Sup{x}\wedge \Sup{y}\]

\subsubsection{The sim gadget}
In the sim gadget player~1 suppose to simulate the run of $M$, and if the simulation is not honest, then player~2 activates the blame gadget.
The simulation of the states is straight forward (since the game graph has states), and the difficulty is to simulate the counter value, more specifically, to simulate the zero testing of the counter.
For this purpose we use the dimensions $r,\ell,g_s$ and $c_+,c_-,g_c$.

We first describe the role of $r,\ell$ and $g_s$.
The reset gadget makes sure that in every invocation of the sim gadget, we have
\[\Avg(g_s) \approx -1, \Avg(r) \approx 1\mbox{ and } \Avg(\ell) \approx 0\]
(The reader should read $a\approx b$ as "the value of $a$ is \emph{very close} to the value of $b$". Precise definitions are given in Section~\ref{sec:DetailedProof})
Then, during the simulation, the blame gadget makes sure that player~1 must play in such way that whenever the machine $M$ is in a right state, then
\[r \approx |g_s| \mbox{ and } \ell \approx 0\]
and whenever the machine is in a left state, then
\[ r \approx 0 \mbox{ and } \ell \approx |g_s|\]  
Intuitively, the role of $\ell$ and $r$ is to make sure that every left to right or right to left transition is simulated by a \emph{significant} number of rounds in the sim gadget, and $g_s$ is a \emph{guard} dimension that make sure that the above assumptions on $r$ and $\ell$ are satisfied.

We now describe the role of $c_+,c_-$ and $g_c$.
In the beginning of each simulation (i.e., every time that the sim gadget is invoked), we have
\[\Avg(c_+) \approx \Avg(c_-) \approx 1 \mbox{ and } \Avg(g_c) \approx -1 \]
During the entire simulation we have $\Avg(g_c) \approx -1$ and if $c$ is the value of the counter in the current simulation (i.e., since the sim gadget was invoked), then
\[c_+ \approx |g_c| + |g_s|c \mbox{ and } c_- \approx |g_c| - |g_s|c\]
Intuitively, whenever $c > 0$, then $c_- \lnapprox |g_c|$, and if $c < 0$ (this can happen only if player~1 is dishonest), then $c_+ \lnapprox |g_c|$.

We now describe the gadgets that simulate the operations of $\INC,\DEC$ and $\NOP$.
The gadgets are illustrated in Figures~\ref{fig:FirstNop}-\ref{fig:Inc} and the following conventions are used:
(i)~Player~1 owns the $\Circle$ vertices, player~2 owns the $\Square$ vertices, and the $\boxbox$ vertex stands for a gadget; (ii)~A transition is labeled either with $a\gets b$ symbol or with a text (e.g., "blame"). For a transition $e$ the label $a\gets b$ stands for $w_a(e) = b$. Whenever the weight of a dimension is not explicitly presented, then the weight is $0$.
We use text labels only to give intuition on the role of the transition. In such transitions the weights of all dimensions are $0$.

In the first state of every $\INC,\DEC$ or $\NOP$ gadget, in a left to right transition, player~1 loops until $\ell \approx 0$ and $r\approx g_s$, and in a right to left transition he loops until $r \approx 0$ and $\ell \approx g_s$.
Then, player~2 decides whether he wants to blame player~1 for violating the assumptions about the values of $\ell,r$ and $g_s$.
\begin{figure}
\begin{minipage}[t]{0.45\linewidth}
     \begin{center} 
    \begin{picture}(25, 20)(0,-20)
\node[Nmarks=i](DEC)(0,0){nop}
\drawloop[loopangle=270](DEC){$
    \begin{array}{{lcrlcr}}
r &\gets& -1, & \ell &\gets& 1 \\
c_+ &\gets& 1, & c_- &\gets& 1 \\ 
g_c &\gets& -1, & & &\\
    \end{array}
    $}

\node[Nmr=0,Nw=12,Nh=12](SIDE)(35,0){side?}
\node[Nw=22,Nh=6,Nmr=3](BLAME)(35,-20){blame $r \to \ell$}
\rmark(BLAME)

\node[Nmr=0,Nw=0,Nh=0](DUMMY)(60,0){}

\drawedge(DEC,SIDE){}
\drawedge(SIDE,BLAME){blame}
\drawedge(SIDE,DUMMY){ok}
\end{picture} 
\end{center}
\caption{nop $r\to\ell$ gadget.}\label{fig:FirstNop}
\end{minipage}
\hfill
\begin{minipage}[t]{0.45\linewidth}     
\begin{center} 
    \begin{picture}(25, 20)(0,-20)
\node[Nmarks=i](DEC)(0,0){nop}
\drawloop[loopangle=270](DEC){$
    \begin{array}{{lcrlcr}}
r &\gets& 1, &\ell &\gets& -1 \\
c_+ &\gets& 1, & c_- &\gets& 1 \\ 
g_c &\gets& -1, & & & \\
    \end{array}
    $}

\node[Nmr=0,Nw=12,Nh=12](SIDE)(35,0){side?}
\node[Nw=22,Nh=6,Nmr=3](BLAME)(35,-20){blame $\ell\to r$}
\rmark(BLAME)

\node[Nmr=0,Nw=0,Nh=0](DUMMY)(60,0){}

\drawedge(DEC,SIDE){}
\drawedge(SIDE,BLAME){blame}
\drawedge(SIDE,DUMMY){ok}

\end{picture} 
     \end{center}
\caption{nop $\ell\to r$ gadget.}
\end{minipage}
\end{figure}


\begin{figure}
\begin{minipage}[t]{0.45\linewidth}
     \begin{center} 
    \begin{picture}(25, 20)(0,-20)



\node[Nmarks=i](DEC)(0,0){dec}
\drawloop[loopangle=270](DEC){$
    \begin{array}{{lcrlcr}}
r &\gets& 1, & \ell &\gets& -1 \\
c_+ &\gets& 0, & c_- &\gets& 2 \\ 
g_c &\gets& -1, & & & \\
    \end{array}
    $}

\node[Nmr=0,Nw=12,Nh=12](SIDE)(35,0){side?}
\node[Nw=22,Nh=6,Nmr=3](BLAME)(35,-20){blame $\ell\to r$}
\rmark(BLAME)

\node[Nmr=0,Nw=0,Nh=0](DUMMY)(60,0){}

\drawedge(DEC,SIDE){}
\drawedge(SIDE,BLAME){blame}
\drawedge(SIDE,DUMMY){ok}

\end{picture} 
\end{center}
\caption{dec $\ell\to r$ gadget.}
\end{minipage}
\hfill
\begin{minipage}[t]{0.45\linewidth}     
\begin{center} 
    \begin{picture}(25, 20)(0,-20)

\node[Nmarks=i](DEC)(0,0){inc}
\drawloop[loopangle=270](DEC){$
    \begin{array}{{lcrlcr}}
r &\gets& -1, & \ell &\gets& 1 \\
c_+ &\gets& 2, & c_- &\gets& 0 \\ 
g_c &\gets& -1, &    &     & \\
    \end{array}
    $}

\node[Nmr=0,Nw=12,Nh=12](SIDE)(35,0){side?}
\node[Nw=22,Nh=6,Nmr=3](BLAME)(35,-20){blame $r \to \ell$}
\rmark(BLAME)

\node[Nmr=0,Nw=0,Nh=0](DUMMY)(60,0){}

\drawedge(DEC,SIDE){}
\drawedge(SIDE,BLAME){blame}
\drawedge(SIDE,DUMMY){ok}

\end{picture} 
     \end{center}
\caption{inc $r\to\ell$ gadget.}\label{fig:Inc}

\end{minipage}
\end{figure}
A transition $q: \IF\ c=0\ \GOTO\ p\ \ELSE\ c:= c - 1\ \GOTO\ p'$, for $q\in Q$ and $p,p'\in P$ is described in Figure~\ref{fig:ZeroTest}.

\begin{figure}
\begin{picture}(125,45)(-30,-10) 



\node(P)(90,23){$p$}

\node(Pprime)(90,3){$p '$}

\node[Nmarks=i](Q)(-30,13){$q$}

\node[Nmr=0,Nw=12,Nh=12](INC_A)(20,23){$c>0$?}

\node[Nw=22,Nh=6,Nmr=3](INC_A_BLAME)(-30,23){blame $c>0$}

\rmark(INC_A_BLAME)

\drawedge[curvedepth=-8,ELside=r](INC_A,INC_A_BLAME){blame}

\drawedge(Q,INC_A){declare $c=0$}

\node[Nw=22,Nh=6,Nmr=3](INC_GAD)(55,23){nop $\ell\to r$}
\rmark(INC_GAD)

\drawedge(INC_A,INC_GAD){ok}

\drawedge(INC_GAD,P){}

\node[Nw=22,Nh=6,Nmr=3](A)(20,3){dec $\ell\to r$}
\rmark(A)

\drawedge[ELside=r](Q,A){declare $c>0$}

\node[Nmr=0,Nw=12,Nh=12](C)(55,3){$c<0$?}

\node[Nw=22,Nh=6,Nmr=3](CdummyD)(-30,3){blame $c<0$}
\rmark(CdummyD)


\drawedge(A,C){}
\drawedge(C,Pprime){ok}
\drawedge[curvedepth=8,ELside=l](C,CdummyD){blame}


\end{picture} 
\caption{$q:$ if $c=0$ then goto $p$ else $c:= c - 1 $ goto $p '$}\label{fig:ZeroTest}
\end{figure}
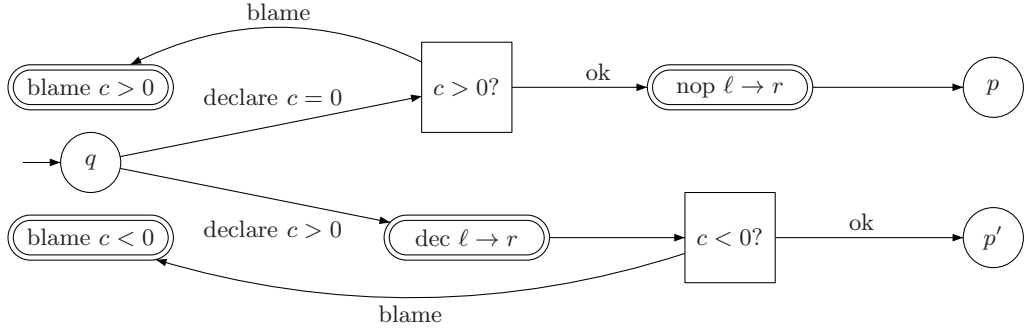

\subsubsection{The blame gadgets}
The role of the blame gadgets is to make sure that the assumptions on $\ell,r$ and $g_s$ are kept in the simulation and to make sure that the zero testing is honestly simulated.
There are four blame gadgets. Two for the honest simulation of $r,\ell$ and $g_s$, and two for the zero testing (one for $c>0$ and one for $c<0$).
The gadgets are described in Figures~\ref{fig:BlameCl0}-\ref{fig:BlameLtR}.
The concept of all four gadgets is similar and hence we describe only the blame gadget for the left to right transition.
We note that in an honest simulation we have $\Avg(r),\Avg(\ell),\Avg(c_+),\Avg(c_-),\Avg(x),\Avg(y) \gtrapprox 0$ in every round.
Hence, if player~1 honestly simulates $M$ and $M$ does not halt, then the winning condition is satisfied.
The blame gadget for the left to right transition is described in Figure~\ref{fig:BlameLtR}.
If the gadget is invoked and $r \lnapprox g_s$, then player~2 can loop on the first state until $\Avg(r) \lnapprox 0$ and still have $\Avg(g_s) \lnapprox 0$.
If $r \approx g_s$, then whenever we have $\Avg(r) \lnapprox 0$ we will also have $\Avg(g_s) \gtrapprox 0$, and thus the winning objective is still satisfied.
If $r \gnapprox g_s$, then we must have $\Avg(\ell) \lnapprox 0$. In this case, player~2 will immediately go to the reset gadget.
We note that player~2 should eventually exit the blame gadget, since otherwise he will lose the game.

\begin{figure}
\begin{minipage}[t]{0.45\linewidth}
     \begin{center} 
    \begin{picture}(25, -20)(0,-10)



\node[Nmarks=i,Nmr=0,Nw=6,Nh=6](TEST)(0,0){}

\drawloop[loopangle=270](TEST){$
    \begin{array}{{lcrlcr}}
c_- &\gets& -1,&
g_c &\gets& 1 \\
    \end{array}
    $}

\node[Nw=22,Nh=6,Nmr=3](RESET)(35,0){reset}
\rmark(RESET)

\drawedge(TEST,RESET){}

\end{picture} 
\end{center}
\caption{blame $c>0$ gadget.}
\end{minipage}
\hfill
\begin{minipage}[t]{0.45\linewidth}     
\begin{center} 
    \begin{picture}(25, -20)(0,-10)

\node[Nmarks=i,Nmr=0,Nw=6,Nh=6](TEST)(0,0){}

\drawloop[loopangle=270](TEST){$
    \begin{array}{{lcrlcr}}
c_+ &\gets& -1,&
g_c &\gets& 1 \\
    \end{array}
    $}

\node[Nw=22,Nh=6,Nmr=3](RESET)(35,0){reset}
\rmark(RESET)

\drawedge(TEST,RESET){}

\end{picture} 
     \end{center}
\caption{blame $c<0$ gadget.}\label{fig:BlameCl0}
\vspace{1cm}
\end{minipage}
\end{figure}


\begin{figure}
\begin{minipage}[t]{0.45\linewidth}
     \begin{center} 
    \begin{picture}(25, 8)(0,-12)

\node[Nmarks=i,Nmr=0,Nw=6,Nh=6](TEST)(0,0){}

\drawloop[loopangle=270](TEST){$
    \begin{array}{{lcrlcr}}
\ell &\gets& -1,&
g_s &\gets& 1 \\
    \end{array}
    $}

\node[Nw=22,Nh=6,Nmr=3](RESET)(35,0){reset}
\rmark(RESET)

\drawedge(TEST,RESET){}

    \end{picture}
\end{center}
  \caption{blame $r\to\ell$ gadget.}
\label{fig:G1}
\end{minipage}
\hfill
\begin{minipage}[t]{0.45\linewidth}     
\begin{center} 
    \begin{picture}(25, 8)(0,-12)

\node[Nmarks=i,Nmr=0,Nw=6,Nh=6](TEST)(0,0){}

\drawloop[loopangle=270](TEST){$
    \begin{array}{{lcrlcr}}
r &\gets& -1,&
g_s &\gets& 1 \\
    \end{array}
    $}

\node[Nw=22,Nh=6,Nmr=3](RESET)(35,0){reset}
\rmark(RESET)

\drawedge(TEST,RESET){}

    \end{picture}
     \end{center}
  \caption{blame $\ell\to r$ gadget.}\label{fig:BlameLtR}
\label{fig:G2}
\end{minipage}
\end{figure}
\subsubsection{The reset gadget}
The role of the reset gadget is to assign the following values for the dimensions:
\[\Avg(\ell) \approx 0, \Avg(r)\approx 1, \Avg(g_s) \approx -1, \Avg(c_-) \approx \Avg(c_+) \approx 1, \Avg(g_c) \approx -1\]
The gadget is described in Figure~\ref{fig:Reset}.
We construct the gadget is such way that each of the players can enforce the above values (player~2 by looping enough time on the first state, player~1 by looping enough time on his two states).
But the construction only gives this option to the players and it does not \emph{punish} a player if he acts differently.
However, the game graph is constructed in such way that if:
\begin{itemize}
\item $M$ does not halt and in the reset gadget (at least one of the players) correctly reset the values, then player~1 wins.
\item $M$ halts and in the reset gadget (at least one of the players) correctly reset the values, then player~2 wins.
\end{itemize}
Hence, if $M$ does halts, then player~2 winning strategy will make sure that the reset assigns correct values, and if $M$ halts, then we can relay on player~1 to reset the values.
We note that player~2 will not stay forever in his state (otherwise he will lose).
In order to make sure that player~1 will not stay forever in one of his states we introduce two \emph{liveness dimensions}, namely $x$ and $y$.
In the simulation and blame gadgets they get $0$ values.
But if player~1 remains forever in one of his two states in the reset gadget, then either $x$ or $y$ will have negative lim-sup value and player~1 will lose.
Hence, in the reset gadget, player~1 should not only reset the values, but also assign a positive value for $y$ and then a positive value for $x$.

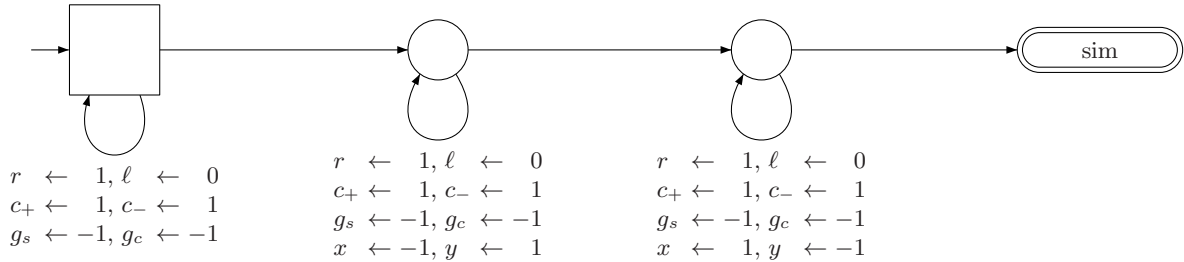
\begin{figure}
\begin{picture}(40,30)(0,-30) 



\node[Nmarks=i,Nmr=0,Nw=12,Nh=12](BAD)(0,0){}

\drawloop[loopangle=270](BAD){$
    \begin{array}{{lcrlcr}}
r &\gets & 1, & \ell &\gets & 0\\
c_+ & \gets & 1, & c_- & \gets & 1\\
g_s & \gets & -1, & g_c & \gets & -1\\
    \end{array}
    $}

\node(GOOD1)(43,0){}
\drawloop[loopangle=270](GOOD1){$
    \begin{array}{{lcrlcr}}
r &\gets & 1,& \ell &\gets & 0\\
c_+ & \gets & 1, &c_- & \gets & 1\\
g_s & \gets & -1,& g_c & \gets & -1\\
x&\gets&-1, & y&\gets&1\\
    \end{array}
    $}

\node(GOOD2)(86,0){}
\drawloop[loopangle=270](GOOD2){$
    \begin{array}{{lcrlcr}}
r &\gets & 1,& \ell &\gets & 0\\
c_+ & \gets & 1, & c_- & \gets & 1\\
g_s & \gets & -1, & g_c & \gets & -1\\
x&\gets&1, & y&\gets&-1\\
    \end{array}
    $}

\node[Nw=22,Nh=6,Nmr=3](SIM)(131,0){sim}
\rmark(SIM)

\drawedge(BAD,GOOD1){}
\drawedge(GOOD1,GOOD2){}
\drawedge(GOOD2,SIM){}

\end{picture} 
\caption{Reset gadget.}\label{fig:Reset}
\end{figure}

\subsubsection{Correctness of the reduction}
We claim that player~1 has a winning strategy if and only if the machine $M$ does not halt.
We first prove that if $M$ halts, then player~2 has a winning strategy (the proof is informal, and a formal proof is given in Section~\ref{sec:DetailedProof}).
The winning strategy for player~2 is as follows:
In the reset gadget make sure that the \emph{reset invariants} are satisfied. This is done by looping the first state of the reset gadget for enough rounds.
In the sim gadget, whenever the \emph{sim invariants} are not fulfilled or whenever player~1 cheats a zero-test, then player~2 invokes the blame gadget.
If the sim invariants are fulfilled, then it must be the case that the game reaches state $q_f$, and in that case player~2 wins.
Otherwise, in each simulation, the guard dimensions have negative average weights, while at least one of the dimensions $\ell,r,c_+$ or $c_-$ has a negative average weight in the blame gadget.
Hence, we get that $\Sup{g_s},\Sup{g_c} < 0$ and $\Inf{\ell} < 0$ or $\Inf{r}< 0$ or $\Inf{c_-} < 0$ or $\Inf{c_+} < 0$.
Hence, the winning condition is not satisfied and player~2 is the winner.

To prove the converse direction we assume that $M$ never halts and describe a winning strategy for player~1.
The winning strategy is to honestly simulate $M$ while keeping the \emph{sim invariants} and the \emph{reset invariants}.
If player~2 never invokes the blame gadget, then the winning condition is trivially satisfied (since the game stays forever in the sim gadget).
Otherwise, after every invocation of the blame gadget, if a \emph{side blame gadget} was invoked, then either the average value of $r$ and $\ell$ is non-negative or the value of the guard dimension $g_s$ is non-negative.
Hence, eventually, we get that $\Inf{r},\Inf{\ell}\geq 0$ or $\Sup{g_s}\geq 0$.
And similarly, we get that $\Inf{c_+},\Inf{c_-} \geq 0$ or $\Sup{g_c}\geq 0$.
Thus, the winning condition is satisfied, and player~1 is the winner.

\section{Detailed Proof}\label{sec:DetailedProof}
In the previous section we accurately described the reduction, and only the proof of the correctness was informal.
In this section we give a precise proof for the correctness of the reduction, namely, we formally describe player-2 winning strategy in the case that $M$ halts (Subsection~\ref{subsec:MHalts}), and player-1 winning strategy in the case that $M$ does not halt (Subsection~\ref{subsec:MDoesNotHalt}).
In Subsection~\ref{subsec:TwoCounter} we extend the reduction to two-counter machine.

In the next subsections we use the next notations and definitions:
A \emph{round} is a round in the game graph (i.e., either player-1 or player-2 move).
A \emph{simulation step} denote all the rounds that are played in a transition gadget (i.e., in a $\NOP$,$\INC$ or $\DEC$ gadget).
The \emph{total number of rounds} is the total number of rounds from the beginning of the play.

\subsection{If $M$ halts, then player~2 is the winner}\label{subsec:MHalts}
In this subsection we assume that $M$ halts.
We denote by $N$ the number of steps after which $M$ halts (for initial counter value $0$) and we denote $\epsilon = \frac{1}{(N+1)^2}$. WLOG we assume that $N > 10$.
The strategy of player~2 in the reset gadget is to achieve the following \emph{reset invariants} (after the play leaves the gadget):
\begin{itemize}
\item $\Avg(g_s), \Avg(g_c)\leq -\frac{1}{2}$
\item $(1- \frac{\epsilon}{4})|g_s| \leq r \leq (1 + \frac{\epsilon}{4})|g_s|$
\item $-\frac{\epsilon}{4}|g_s| \leq \ell \leq \frac{\epsilon}{4}|g_s|$
\item $(1-\frac{\epsilon}{4})|g_c| \leq c_+,c_- \leq (1 +\frac{\epsilon}{4})|g_c|$
\end{itemize}
We note that player~2 can maintain the above by looping sufficiently long time in the first state, and once the invariants are reached, player~1 cannot violate them in his states in the reset gadget (since the average value of $g_s$ and $g_c$ can only get closer to $-1$, the average value of $\ell$ only gets closer to $0$ and the average value of $r,c_-$ and $c_+$ only gets closer to $1$).

The strategy of player~2 in the sim gadget is to maintain, in every step of the simulation, the next two invariants, which we denote by the \emph{left right invariants}:
\begin{itemize}
\item (Left state invariant) If the machine is in a left state, then 
$(1- \epsilon)|g_s| \leq \ell \leq (1 + \epsilon)|g_s|$ and 
$-\epsilon|g_s| \leq r \leq \epsilon|g_s|$.
\item (Right state invariant) If the machine is in a right state, then
$(1- \epsilon)|g_s| \leq r \leq (1 + \epsilon)|g_s|$ and 
$-\epsilon|g_s| \leq \ell \leq +\epsilon|g_s|$
\end{itemize}
We denote $\delta = \frac{1}{\frac{1}{2} + N(1+2\epsilon)}$.
We first prove that under these invariants $\Avg(g_s) \leq -\delta$.
Then we use this fact to show that if player~1 violates these invariants, then player~2 can violate $(\Inf{\ell} \wedge \Inf{r} \vee\Sup{g_s})$, and therefore he wins.
\begin{lem}\label{lem:lfInvariantsImpAvgGs}
Suppose that after the reset gadget $\Avg(g_s)\leq -\frac{1}{2}$, then as long as the left right invariants are maintained, it always hold that $\Avg(g_s)\leq -\delta$ (provided that the sim gadget does not simulate more than $N$ steps).
\end{lem}
\begin{proof}
We denote by $R$ the number of rounds that were played before the current invocation of the simulation gadget.
We claim that after simulating $i$ steps of the machine (in the current invocation of the sim gadget), the \textbf{total} number of rounds in the play (i.e., number of rounds from the beginning of the play, not from the beginning of the current invocation) is at most $R + i\cdot |g_s|(1+2\epsilon)$.
The proof is by a simple induction, and for the base case $i=0$ the proof is trivial.
For $i > 0$, we assume WLOG that the $i$-th transition is a left-to-right transition.
Hence, before the last simulation step we had $r\geq -\epsilon |g_s|$ and after the $i$-th step was completed we have $r\leq (1+\epsilon)|g_s|$.
Since in every round the value of $r$ is incremented by $1$, we get that at most $(1+2\epsilon)|g_s|$ rounds were played and
the proof of the claim follows (and the proof for a right-to-left transition is symmetric).

Hence, after $N$ simulation steps we have
\[\Avg(g_s) \leq \frac{g_s}{R+N|g_s|(1+2\epsilon)}\]
Since in the beginning of the sim gadget we had $\Avg(g_s) \leq -\frac{1}{2}$, then $R \leq \frac{|g_s|}{2}$. Hence, and since $g_s < 0$ we get
\[\Avg(g_s) \leq \frac{g_s}{\frac{|g_s|}{2}+N|g_s|(1+2\epsilon)} = - \frac{1}{\frac{1}{2} + N(1+2\epsilon)} = -\delta \]

We note that unless the blame gadget is invoked, the value of $g_s$ is not changed in the simulation gadget. Hence, $\Avg(g_s)$ gets the maximal value after the $N$-th step and the proof is complete.
\qed
\end{proof}

\begin{lem}\label{lem:lfInvariantsImpLf}
Let $\gamma = \min(\frac{\epsilon\delta}{2},\frac{\frac{\epsilon}{4}}{1 + \frac{1}{\delta} - \frac{\epsilon}{4}})$.
If player~1 violates the left-right invariants in the first $N$ steps of a simulation, then player~2 can achieve in the blame gadget either $\Avg(r) \leq -\gamma$ or $\Avg(\ell) \leq -\gamma$ (or both) while maintaining $\Avg(g_s),\Avg(g_c)\leq -\gamma$.
\end{lem}
\begin{proof}
We prove the assertion for the left-state invariant and the proof for the right-state is symmetric.
Recall that the invariant consistences of four assumptions, namely, (i)~$(1- \epsilon)|g_s| \leq \ell$; (ii)~$\ell \leq (1 + \epsilon)|g_s|$; (iii)~$-\epsilon|g_s| \leq r$; and (iv)~$r \leq \epsilon|g_s|$.
We first prove the assertion when the first condition is violated, i.e., we assume that $\ell < (1 -\epsilon)|g_s|$.
If this is the case after a right-to-left transition, then player~2 will invoke the blame $r\to\ell$ gadget after the transition ends.
In the blame gadget he will traverse the self-loop for $X\cdot(1 - \frac{\epsilon}{2})$ times, where $X$ is the value of $|g_s|$ before the invocation of the blame gadget, and then he will go to the reset gadget.
As a result (since in every loop $\ell$ is decremented by $1$ and $g_s$ is incremented by $1$) we get that the value of $\ell$ and $g_s$ is at most $-X\cdot\frac{\epsilon}{2}$.
By Lemma~\ref{lem:lfInvariantsImpAvgGs} we know that before the invocation of the blame gadget we had $\Avg(g_s) \leq -\delta$.
Hence, if $R$ is the number of rounds before the invocation of the blame gadget, then $R \leq \frac{X}{\delta}$.
Hence, after the blame gadget ends, we have 
\[ \Avg(\ell),\Avg(g_s) \leq -\frac{X\cdot\frac{\epsilon}{2}}{R + X\cdot(1 - \frac{\epsilon}{2})} \leq -\frac{X\cdot\frac{\epsilon}{2}}{\frac{X}{\delta} + X\cdot(1 - \frac{\epsilon}{2})} = -\frac{\frac{\epsilon}{2}}{1 + \frac{1}{\delta} - \frac{\epsilon}{2}}\]

If the second condition is violated, namely, if $\ell > (1 + \epsilon)|g_s|$, then we claim that it must be the case that $r < -\frac{\epsilon|g_s|}{2}$.
Indeed, when the sim gadget is invoked we have $r \leq |g_s|(1 +\frac{\epsilon}{4})$ and $\ell \leq |g_s|\frac{\epsilon}{4}$.
In the sim gadget the value of the sum $r + \ell$ is not changed (since $r$ is incremented only when $\ell$ is decremented and vise versa).
Hence, the sum never exceeds $|g_s|( 1 + \frac{\epsilon}{2})$.
Thus, if $\ell > (1 + \epsilon)|g_s|$, then it must be the case that $r < -\frac{\epsilon|g_s|}{2}$ and we get that $\Avg(r) \leq -\frac{\epsilon}{2} \delta$.
Hence, by leaving the blame gadget after two rounds we get $\Avg(r) \leq -\frac{\epsilon\delta}{2}$ while $\Avg(g_s)\leq -\delta$.

If the third condition is violated, namely, if $r < -\epsilon|g_s|$, then $\Avg(r) = -\epsilon|\Avg(g_s)| \leq -\epsilon\delta$.
Hence, if player~2 can choose to exit the blame gadget after two rounds and we get $\Avg(r) \leq -\epsilon\delta$ while $\Avg(g_s)\leq -\delta$.

Finally, if the fourth condition is violated, namely, if $r > \epsilon|g_s|$, then by analyzing the sum $r+ \ell$ we get that $\ell \leq (1 - \frac{\epsilon}{2})|g_s|$.
We repeat the same analysis as in the case where the first invariant is violated (i.e., when $\ell \leq (1 - \epsilon)|g_s|$) and get that
\[\Avg(g_s),\Avg(r) \leq 
-\frac{\frac{\epsilon}{4}}{1 + \frac{1}{\delta} - \frac{\epsilon}{4}}\]

It is an easy observation that if at beginning of the last invocation of the sim gadget $\Avg(g_c) \leq -\frac{1}{2}$, then it remains at most $-\frac{1}{2}$ as it gets a value of $-1$ in every round.
The proof is complete. \qed
\end{proof}
By Lemma~\ref{lem:lfInvariantsImpLf}, if player~2 maintains the reset invariant in the reset gadget, then in every simulation player~1 must satisfy the left-right invariants. Otherwise, we get that infinitely often the average value of either $r$ or $\ell$ is $-\gamma$ while the average value of $g_s$ is always at most $-\gamma$. Hence $\Sup{g_s} < 0$ and either $\Inf{r} < 0$ or $\Inf{\ell} < 0$ and thus the condition $(\Inf{\ell} \wedge \Inf{r} \vee\Sup{g_s})$ is violated and therefore player~1 is losing.

In the next three lemmas we prove that player~1 must honestly simulates the zero-testing.
The first lemma is a simple corollary from the left-right invariants.
\begin{lem}\label{lem:LRNumRounds}
Under the left-right invariants, in the $\DEC$,$\INC$ and $\NOP$ gadgets, player~1 follows the self-loop of the first state at most $|g_s|(1+2\epsilon)$ rounds and at least $|g_s|(1-2\epsilon)$ rounds.
\end{lem}
\begin{proof}
Upper bound:
in a left to right transition, we have $\ell \leq |g_s|(1+\epsilon)$ before the transition and $\ell \geq -|g_s|\epsilon$ after the transition, and in every round $\ell$ is decremented by $1$.
Lower bound:
in a left to right transition, we have $\ell \geq |g_s|(1-\epsilon)$ before the transition and $\ell \leq |g_s|\epsilon$ after the transition, and in every round $\ell$ is decremented by $1$.

The proof for a right to left transition is symmetric.
\end{proof}
The next lemma shows the correlation between $g_c$ and $c_+$ and $c_-$.
\begin{lem}\label{lem:CorrelationOfCountersDim}
Let $\#\INC$ (resp., $\#\DEC$) be the number of times that the $\INC$ ($\DEC$) gadget was visited, and we denote $c = \#\INC - \#\DEC$ (namely, $c$ is the actual value of the counter in the counter machine $M$).
Then under the left-right invariants, in the first $N$ steps of the simulation we always have $c_+ \leq |g_c|(1+\epsilon) + c|g_s| + \frac{|g_s|}{2}$ and $c_- \leq |g_c|(1+\epsilon) - c|g_s| + \frac{|g_s|}{2}$.
\end{lem}
\begin{proof}
We prove the claim of the lemma for $c_+$ and the proof for $c_-$ is symmetric.
Let $X$ be the value of $|g_c|$ when the sim gadget is invoked.
By the reset invariants we get that $c_+\leq X(1+\frac{\epsilon}{4})$.
By Lemma~\ref{lem:LRNumRounds} we get that every visit in the $\INC$ gadget contributes at most $|g_s|(1+2\epsilon)$ more to $c_+$ than its contribution to $g_c$ and every visit in the $\DEC$ contributes at least $|g_s|(1-2\epsilon)$ more to to $g_c$ than its contribution to $c_+$.
Hence,
\begin{quote}
$c_+ \leq X(1+\frac{\epsilon}{4}) + (|g_c| - X) + \#\INC\cdot |g_s|(1+2\epsilon) - \#\DEC \cdot |g_s|(1-2\epsilon) =$\\*
$|g_c| + \epsilon X + (\#\INC - \#\DEC)|g_s|(1+2\epsilon) + 4\epsilon|g_s|\cdot\#\DEC$
\end{quote}

We recall that $c = (\#\INC - \#\DEC)$, and observe that $X \leq |g_c|$, and that $\#\DEC \leq N$ and thus $\epsilon \cdot\#\DEC < \frac{1}{10}$.
Hence, we get that $c_+ \leq |g_c|(1+\epsilon) + c|g_s| + \frac{|g_s|}{2}$.
\qed
\end{proof}
The next lemma suggest that player~1 must honestly simulates the zero-tests.
\begin{lem}\label{lem:ZeroTestMust}
If the reset and left-right invariants hold, then for $\gamma = \min(\frac{1}{20N},\frac{\delta}{8})$ the following hold: (i)~if the blame $c<0$ is invoked and $c<0$ then player~2 can achieve $c_+ \leq -\gamma$ while maintaining $\Avg(g_s),\Avg(g_c) \leq -\gamma$; and
(ii)~if the blame $c>0$ is invoked and $c>0$ then player~2 can achieve $c_- \leq -\gamma$ while maintaining $\Avg(g_s),\Avg(g_c) \leq -\gamma$.
\end{lem}
\begin{proof}
We prove the first item of the lemma and the proof for the second item is symmetric.
Suppose that $c<0$ (i.e., $c\leq -1$) when blame $c<0$ gadget is invoked.
Let $X$ and $Y$ be the values of $|g_c|$ and $|g_s|$ before the invocation of the blame gadget. 
Then by Lemma~\ref{lem:CorrelationOfCountersDim}, before the invocation we have $c_+ \leq X(1+\epsilon) - \frac{Y}{2}$.
Hence, by traversing the self-loop in the first state of the blame $c<0$ gadget for $X(1 + \epsilon) - \frac{Y}{4}$ rounds we get $c_+ \leq - \frac{Y}{4}$ and $g_c \leq \epsilon X - \frac{Y}{4}$.
Let $R$ be the number of rounds that were played from the beginning of the play (and not just from the beginning of the current invocation of the sim gadget).
Since $g_c$ is decremented by at most $1$ in every round we get that $X(1 + \epsilon) - \frac{Y}{4} \leq 2X \leq 2R$.
By lemma~\ref{lem:lfInvariantsImpAvgGs} we have $\frac{Y}{R} \leq -\delta$.
Hence, $\Avg(c_+) \leq \frac{c_+}{2R} \leq -\frac{Y}{8R} \leq -\frac{\delta}{8}$.
Similarly, since $\frac{X}{R}$ is bounded by $1$, we have $\Avg(g_c) \leq \frac{\epsilon X}{2R} - \frac{\delta}{8} \leq \frac{\epsilon}{2} - \frac{\delta}{8}$.
Recall that $\delta = \frac{1}{\frac{1}{2} + N(1+2\epsilon)}$.
Hence, $\Avg(g_c)\leq \frac{2\epsilon + 4N\epsilon + 8\epsilon^2 - 1}{8(\frac{1}{2} + N(1+2\epsilon))}$ and since $\epsilon = \frac{1}{(N+1)^2}$ and $N > 10$ we get that $\Avg(g_c) \leq -\frac{1}{20N}$.
The value of $g_s$ was at most $-\delta R$ before the blame gadget, and in the blame gadget $g_s$ is decreased by $1$ in every round. Hence $\Avg(g_s) \leq -\delta$ and the proof follows by taking $\gamma = \min(\frac{1}{20N},\frac{\delta}{8})$.
\qed
\end{proof}

We are now ready to prove one side of the reduction.
\begin{prop}\label{prop:IfHaltsThenWin}
If the counter machine $M$ halts, then player~2 has a winning strategy for violating $(\Inf{\ell} \wedge \Inf{r} \vee\Sup{g_s})\wedge (\Inf{c_+}\wedge \Inf{c_-} \vee \Sup{g_c})\wedge \Sup{x}\wedge \Sup{y}$.
Moreover, if $M$ halts then there exists a constant $\zeta > 0$ that depends only on $M$ such that player~2 has a winning strategy for violating
$(\Inf{\ell} \geq -\zeta \wedge \Inf{r} \geq -\zeta  \vee\Sup{g_s} \geq -\zeta)\wedge (\Inf{c_+} \geq -\zeta \wedge \Inf{c_-}\geq -\zeta \vee \Sup{g_c} \geq -\zeta)\wedge \Sup{x}\geq -\zeta \wedge \Sup{y} \geq -\zeta$.
\end{prop}
\begin{proof}
Suppose that $M$ halts and let $N$ be the number of steps that $M$ runs before it halts (for an initial counter value $0$).
Player-2 strategy is to
\begin{itemize}
\item Maintain the reset-invariants.
\item Whenever the left-right invariants are violated, he invokes the blame $r\to\ell$ or $\ell\to r$ gadget.
\item Whenever the zero-testing is dishonest, he activates the corresponding blame gadget (either $c>0$ or $c<0$).
\item If $q_f$ is reached, he stays there forever.
\end{itemize}
The correctness of the construction is immediate by the lemmas above.
We first observe that it is possible for player~2 to satisfy the reset-invariants and that if player~1 stays in the reset gadget forever, then he loses.

Whenever the left-right invariant is violated, then the average weight of $r$ and/or $\ell$ is negative, while the average weight of $g_s$ and $g_c$ remains negative.
Hence, if in every simulation player~1 violates the left-right invariants in the first $N$ steps we get that the condition is violated since $\Sup{g_s} \leq -\gamma$ and either $\Inf{r} \leq -\gamma$ or $\Inf{\ell} \leq -\gamma$.
Hence, we may assume that these invariants are kept in every simulation.

Whenever the zero-testing is dishonest (while the left-right invariants are satisfied), then by Lemma~\ref{lem:ZeroTestMust}, player~2 can invoke a counter blame gadget and achieve negative average for either $c_+$ or $c_-$ while maintaining $g_c$ and $g_s$ negative.
If in every simulation player~1 is dishonest in zero-testing, then we get that either $\Inf{c_-} \leq -\gamma$ or $\Inf{c_+} \leq -\gamma$ while $\Sup{g_c} \leq -\gamma$ and the condition is violated.
Hence, we may assume that player~1 honestly simulates the zero-tests.
Finally, if the transitions of $M$ are properly simulated, then it must be the case the state $q_f$ is reached and when looping this state forever player~1 loses (since $\Sup{x} \leq -1 < 0$).
\qed
\end{proof}
\subsection{If $M$ does not halt, then player~1 is the winner}\label{subsec:MDoesNotHalt}
Suppose that $M$ does not halt.
The strategy of player~1 in the reset gadget is as following:
Let $i$ be the number of times that the reset gadget was visited, and we denote $\epsilon_i = \frac{1}{i+10}$.
Similarly to player-2 strategy in Subsection~\ref{subsec:MHalts}, player-1 strategy in the reset gadget is to achieve the following invariants (after the play leaves the gadget):
\begin{itemize}
\item $\Avg(g_s), \Avg(g_c)\leq -\frac{1}{2}$
\item $(1- \frac{\epsilon_i}{4})|g_s| \leq r \leq (1 + \frac{\epsilon_i}{4})|g_s|$
\item $-\frac{\epsilon_i|g_s|}{4} \leq \ell \leq \frac{\epsilon_i|g_s|}{4}$
\item $(1-\frac{\epsilon_i}{4})|g_c| \leq c_+,c_- \leq (1 +\frac{\epsilon_i}{4})|g_c|$
\end{itemize}
To satisfy these invariants, he follows the self-loop of his first state until $\Avg(y) \geq 0$ and then follows the self-loop of the second state until the invariants are fulfilled and $\Avg(x) \geq 0$.
In the sim gadget, player-1 strategy is to simulate every $\NOP$,$\INC$ and $\DEC$ step by following the self-loop in the corresponding gadget for $|g_s|$ rounds. In addition, player~1 honestly simulates the zero-tests.

We denote the above player-1 strategy by $\tau$ and next lemma shows the basic properties of a play according to $\tau$.
\begin{lem}\label{lem:PropsPlayer1Stra}
In any play according to $\tau$, after the reset gadget was visited for $i$ times, in the sim gadget we always have:
(i)~In a right state: $r\geq -\epsilon_i|g_s|, \ell \geq (1-\epsilon_i)|g_s|$ and in a left state $\ell \geq -\epsilon_i|g_s|, r \geq (1-\epsilon_i)|g_s|$; and (ii)~$c_+ \geq (1-\epsilon_i)|g_c| + c|g_s|$ and $c_-\geq (1-\epsilon_i)|g_c| - c|g_s|$, where $c = \#\INC - \#\DEC$ in the current invocation of the sim gadget.
\end{lem}
\begin{proof}
The proof of the first item is straight forward.
Initially, when the sim gadget is invoked (and starts in a left state) we have $-\epsilon_i |g_s| \leq \ell \leq \epsilon_i |g_s|$ and $(1-\epsilon_i) |g_s| \leq r \leq (1+\epsilon_i) |g_s|$.
In every left to right transition $r$ is decremented by $|g_s|$ and $\ell$ is incremented by $|g_s|$.
In every right to left transition $\ell$ is decremented by $|g_s|$ and $r$ is incremented by $|g_s|$.
Hence, the proof of the first item follows.

In order to prove the second item we denote by $K$ the value of $|g_c|$ when the sim gadget is invoked.
Initially we have $c_+ \geq (1-\epsilon_i)K$.
After every simulation step we have $c_+ \geq (1-\epsilon_i)K + |g_s|(2\#\INC + \#\NOP)$ and $|g_c| = K + |g_s|(\#\INC + \#\DEC + \#\NOP)$.
Hence, $c_+ \geq (1-\epsilon_i)|g_c| + c|g_s|$.
The proof for $c_-$ is symmetric.
\qed
\end{proof}

We will use the next lemma to prove that player~2 is losing when he invokes the blame gadgets.
\begin{lem}\label{lem:Stra1Blame}
In a play prefix consistent with $\tau$:
\begin{enumerate}
\item In the blame $r\to\ell$ gadget: whenever $\Avg(\ell)\leq -\epsilon_i$, then $\Avg(g_s)\geq -\epsilon_i$.
\item In the blame $\ell\to r$ gadget: whenever $\Avg(r)\leq -\epsilon_i$, then $\Avg(g_s)\geq -\epsilon_i$.
\item In the blame $c < 0$ gadget: whenever $\Avg(c_+) \leq -\epsilon_i$, then $\Avg(g_c) \geq -\epsilon_i$.
\item In the blame $c > 0$ gadget: whenever $\Avg(c_-) \leq -\epsilon_i$, then $\Avg(g_c) \geq -\epsilon_i$.
\end{enumerate}
Where $i$ is the number of times that the reset gadget was visited.
\end{lem}
\begin{proof}
Proof of item 1:
By Lemma~\ref{lem:PropsPlayer1Stra} we know that when blame $r\to\ell$ gadget is invoked, the value of $\ell$ is at least $(1-\epsilon_i)|g_s|$.
Let $R$ be the total number of rounds that were played before the blame gadget was invoked, and by $Y$ the value of $|g_s|$ before the invocation of the blame gadget.
In order to obtain $\Avg(\ell)\leq -\epsilon_i$, player~1 should follow the self-loop of the gadget at least $X=\frac{(1-\epsilon_i)Y+R\epsilon_i}{1-\epsilon_i}$ rounds.
In every such round the value of $g_s$ is incremented by $1$.
We note that $X\geq Y$.
Hence, after $X$ rounds we have $g_s = -Y + X \geq 0$ and in particular $\Avg(g_s)\geq -\epsilon_i$.

Proof of item 2: Symmetric to the proof of item 1.

Proof of item 3:
By Lemma~\ref{lem:Stra1Blame} we know that before the invocation of the blame gadget we have $c_+ \geq (1-\epsilon_i)|g_c| + c|g_s|$.
Since $\tau$ honestly simulates the zero-tests we get that $c\geq 0$.
Hence, $c_+ \geq (1-\epsilon_i)|g_c|$.
Thus, by the same arguments as in the proof of item 1 we get that if $\Avg(c_+)\leq -\epsilon_i$, then $\Avg(g_c) \geq 0 \geq -\epsilon_i$.

Proof of item 4:
By Lemma~\ref{lem:Stra1Blame} we know that before the invocation of the blame gadget we have $c_- \geq (1-\epsilon_i)|g_c| - c|g_s|$.
Since $\tau$ honestly simulates the zero-tests we get that $c = 0$.
The rest of the proof is symmetric to the proof of item 3.
\qed
\end{proof}

We are now ready to prove the $\tau$ is a winning strategy.
\begin{prop}\label{prop:IfNotHaltOneWin}
If $M$ does not halt, then $\tau$ is a winning strategy.
\end{prop}
\begin{proof}
In order to prove that $\tau$ satisfies the condition $(\Inf{\ell} \wedge \Inf{r} \vee\Sup{g_s})\wedge (\Inf{c_+}\wedge \Inf{c_-} \vee \Sup{g_c})\wedge \Sup{x}\wedge \Sup{y}$ it is enough to prove that when playing according to $\tau$, for any constant $\delta > 0$ the condition
$(\Inf{\ell}\geq -\delta \wedge \Inf{r} \geq -\delta \vee\Sup{g_s} \geq -\delta )\wedge (\Inf{c_+}\geq -\delta \wedge \Inf{c_-}\geq -\delta \vee \Sup{g_c}\geq -\delta)\wedge \Sup{x}\wedge \Sup{y}$ is satisfied.

Let $\delta > 0$ be an arbitrary constant and in order to prove the claim we consider two distinct cases.

In the first case, player~2 strategy will invoke the blame gadgets only finitely many times. Hence, either there is a suffix that is played only in a blame gadget or only in a reset gadget (and in such suffixes player~2 loses) or
there is an infinite suffix that is played only in the sim gadget (as $M$ does not halt).
In the sim gadget the values of $x,y$ and $g_s$ are not changed. Hence the long-run average weight of these dimensions is $0$.
In addition, $c_+$ and $c_-$ are never decremented in the sim gadget.
Hence, their long-run average weight is also at least $0$.
Thus, the condition is satisfied.

In the second case we consider, player~2 always eventually invokes a blame gadget.
Since a blame gadget is invoked infinitely many times we get that the reset gadget is invoked infinitely often, and thus $\Sup{x},\Sup{y} \geq 0$.
In addition, the sim gadget is invoked infinitely often.
Let $i$ be the minimal index for which $\epsilon_i \leq \delta$.
We claim that after the $i$-th invocation of the sim gadget, in every round (i)~either $\Avg(\ell)\geq -\epsilon_i \wedge \Avg(r)\geq -\epsilon_i$ or $\Avg(g_s)\geq -\epsilon_i$; and (ii)~either $\Avg(c_+)\geq -\epsilon_i \wedge \Avg(c_+)\geq -\epsilon_i$ or $\Avg(g_c) \geq -\epsilon_i$.
The proof for the first item follows from the reset invariants, from Lemma~\ref{lem:PropsPlayer1Stra} (as $r,\ell \geq -\epsilon_i |g_s|$ and $|\Avg(g_s)| \leq 1$) and from Lemma~\ref{lem:Stra1Blame} (and from the fact that $r$ and $\ell$ are never decremented in the reset gadget).
The proof for the second item follows from the reset invariants, from the fact that $c_+$ and $c_-$ are never decremented in the sim gadget and from Lemma~\ref{lem:Stra1Blame} (and from the fact that $c_+$ and $c_-$ are never decremented in the reset gadget).
Thus, from certain round, either $\Avg(\ell)$ and $\Avg(r)$ are always at least $-\delta$, or infinitely often $\Avg(g_s) \geq -\delta$.
Hence, $(\Inf{\ell}\geq -\delta \wedge \Inf{r} \geq -\delta \vee\Sup{g_s} \geq -\delta )$ is satisfied and similarly $(\Inf{c_+}\geq -\delta \wedge \Inf{c_-}\geq -\delta \vee \Sup{g_c}\geq -\delta)$ is satisfied.
Therefore, since $\epsilon_i \leq \delta$, we get that $\tau$ satisfies $(\Inf{\ell}\geq -\delta \wedge \Inf{r} \geq -\delta \vee\Sup{g_s} \geq -\delta )\wedge (\Inf{c_+}\geq -\delta \wedge \Inf{c_-}\geq -\delta \vee \Sup{g_c}\geq -\delta)\wedge \Sup{x}\wedge \Sup{y}$.
\qed
\end{proof}
\subsection{Extending the reduction to two-counter machine}\label{subsec:TwoCounter}
When $M$ is a two-counter machine, we use 4 dimensions for the counters, namely $c_+^1,c_-^1,c_+^2,c_-^2$ and one guard dimension $g_c$.
The winning condition is $(\Inf{\ell} \wedge \Inf{r} \vee\Sup{g_s})\wedge (\Inf{c_+^1}\wedge \Inf{c_-^1}\wedge\Inf{c_+^2}\wedge \Inf{c_-^2} \vee \Sup{g_c})\wedge \Sup{x}\wedge \Sup{y}$.
In a $\NOP$ gadget all four dimensions $c_+^1,c_-^1,c_+^2,c_-^2$ get a value of $1$ in the self-loop.
When a counter $c_i$ (for $i=1,2$) is incremented (resp., decremented), then counter $c_+^i$ and $c_-^i$ are assigned with weights according to the weights of $c_+$ and $c_-$ in the $\INC$ ($\DEC$) gadget that we described in the reduction for a one counter machine, and $c_+^{3-i},c_-^{3-i}$ are assigned with weights according to a $\NOP$ gadget.

The proofs of Proposition~\ref{prop:IfHaltsThenWin} and Proposition~\ref{prop:IfNotHaltOneWin} easily scale to a two-counter machine.
Hence, the undecidability result is obtained.
\begin{thm}\label{thm:Undec}
The problem of deciding the winner in a robust multidimensional mean-payoff game with ten dimensions is undecidable.
\end{thm}

\section{Discussion}\label{sec:Discuss}
\Heading{Additional intuition.}
Previous undecidability results for mean-payoff games relied on a reduction from the universality problem of a non-deterministic weighted automaton over finite-words.
In such setting, the input for the automaton is a run of a counter machine and the automaton verifies that the run is valid.
The non-determinism allows the automaton to decide after the run ends whether (i)~it simulates the first counter or the second counter; and (ii)~whether a dishonest zero-test was done for a positive counter or for a zero counter.

In our reduction player~2 must identify dishonest simulation during the run and he does not know in advance whether a dishonest zero-test will be done for the first counter or for the second counter.
Hence, if we would apply the standard technique, then one run that is dishonest for the first counter might \emph{fix} the mean-payoff of a previous run that was dishonest for the second counter.
For this purpose we introduce the reset gadget that in a way \emph{clears} the affect of previous simulations.
However, with the reset gadget the affect of the rounds that are played in the simulation gadget becomes neglectable.
For this purpose we introduce the notion of left-to-right and right-to-left transitions who make sure that every simulation step takes a non-neglectable fraction of the run.
Both the reset gadget and the blame gadgets crucially relies on the combination of limit-supremum-average and limit-infimum-average objectives.
Indeed, games over the condition $(\Inf{\ell} \wedge \Inf{r} \vee\Inf{g_s})\wedge (\Inf{c_+^1}\wedge \Inf{c_-^1}\wedge\Inf{c_+^2}\wedge \Inf{c_-^2} \vee \Inf{g_c})\wedge \Inf{x}\wedge \Inf{y}$ or the condition
$(\Sup{\ell} \wedge \Sup{r} \vee\Sup{g_s})\wedge (\Sup{c_+^1}\wedge \Sup{c_-^1}\wedge\Sup{c_+^2}\wedge \Sup{c_-^2} \vee \Sup{g_c})\wedge \Sup{x}\wedge \Sup{y}$ are decidable~\cite{VelnerR11}.

\Heading{Undecidability for Similar Objectives}
Alur et al~\cite{Alur:2009:ODM:1532848.1532880} considered the same objectives as in this work.
Boker et al.~\cite{BokerCHK11} and Tomita et al.~\cite{TomitaHHY12} extended Alur et al objectives with boolean temporal logic. Hence, the undecidability result trivially holds for their model.

Chatterjee et al.~\cite{mean-payoff-Automaton-Expressions} considered a quantitative objective that assigns a (one-dimensional) real value to the vectors $\LimInfAvg(\pi)$ and $\LimSupAvg(\pi)$.
In their framework an atomic expression is $\LimInfAvg_i$ and $\LimSupAvg_i$ and a \emph{mean-payoff expression} is the closure under the $\min,\max$ and sum operators and under the numerical complement (multiplication by $-1$).
For example, the value of the expression $E = \LimInfAvg_1 + \max(\LimInfAvg_2,\LimSupAvg_3)$ for an infinite path $\pi$ is $\LimInfAvg_1(\pi) + \max(\LimInfAvg_2(\pi),\LimSupAvg_3(\pi))$.
As the objective is quantitative and not boolean, there are three problem that are relevant to games for a given expression $E$:
(i)~Can player~1 assure $E\geq \nu$?; (ii)~Can player~1 assure $E > \nu$?; and (iii)~What is the maximal (supremum) value the player~1 can assure?
In~\cite{Velner14} we proved that when player~1 is restricted to finite-memory strategies, then the decidability of the first problem is inter-reducible with Hilbert's Tenth problem over rationals (a long standing open problem),
but the second problem is decidable and the third problem is computable.
For arbitrary (infinite-memory) strategies we can easily extend our undecidability proof from Section~\ref{sec:Reduction} and~\ref{sec:DetailedProof} and get the next theorem:
\begin{thm}\label{thm:MPExpAreUndec}
Two-player games over mean-payoff expression are undecidable.
\end{thm}
\begin{proof}
Consider the construction from Section~\ref{sec:Reduction} and the expressions
\\$E_1 = \max(\min(\LimInfAvg_\ell,\LimInfAvg_r),\LimSupAvg_{g_s})$,
\\$E_2 = \max(\min(\LimInfAvg_{c_+^1},\LimInfAvg_{c_-^1},\LimInfAvg_{c_+^2},\LimInfAvg_{c_-^2}),\LimSupAvg_{g_c})$,
$E_3 = \min(\LimSupAvg_x,\LimSupAvg_y)$, $E = \min(E_1,E_2,E_3)$ and $F = -E$,
and the boolean formula $\varphi = (\Inf{\ell} \wedge \Inf{r} \vee\Sup{g_s})\wedge (\Inf{c_+^1}\wedge \Inf{c_-^1}\wedge\Inf{c_+^2}\wedge \Inf{c_-^2} \vee \Sup{g_c})\wedge \Sup{x}\wedge \Sup{y}$.
It is a trivial observation that for any play $\pi$:
(i)~$\varphi$ is satisfied iff $E(\pi)\geq 0$; and (ii)~$\varphi$ is not satisfied iff $F(\pi)>0$.
Hence, we have a reduction from determining whether player~1 wins for robust multidimensional mean-payoff games to mean-payoff expression games with weak inequality, and a reduction from determining whether player~2 wins for robust multidimensional mean-payoff games to mean-payoff expression games with strong inequality.
Moreover, by Proposition~\ref{prop:IfHaltsThenWin} it follows that if player~1 cannot satisfy $E \geq 0$, then for some $\zeta > 0$ he cannot satisfy $E \geq -\zeta$.
Hence if we could compute the maximal (supremum) value that player~1 can assure we could determine whether he wins for $E \geq 0$ (if the value is negative, then he loses, and otherwise he wins). \qed
\end{proof}

\bibliography{mpbibtex}{}
\bibliographystyle{plain}

\end{document}